\documentclass[pdflatex,sn-basic]{sn-jnl}

\jyear{2021}%

\theoremstyle{thmstyleone}%
\newtheorem{theorem}{Theorem}
\newtheorem{lemma}{Lemma}

\theoremstyle{thmstyletwo}%
\newtheorem{remark}{Remark}%

\theoremstyle{thmstylethree}%
\newtheorem{definition}{Definition}%

\begin{document}

\title[Withdrawal Success Estimation]{Withdrawal Success Estimation}


\author[1]{\fnm{Hayden} \sur{Brown} ORCID: 0000-0002-2975-2711}\email{haydenb@nevada.unr.edu}

\affil[1]{\orgdiv{Department of Mathematics and Statistics}, \orgname{University of Nevada, Reno}, \orgaddress{1664 \street{N. Virginia Street}, \city{Reno}, \postcode{89557}, \state{Nevada}, \country{USA}}}


\abstract{Given a geometric L\'evy alpha-stable wealth process, a log-L\'evy alpha-stable lower bound is constructed for the terminal wealth of a regular investing schedule. Using a transformation, the lower bound is applied to a schedule of withdrawals occurring after an initial investment. As a result, an upper bound is described on the probability to complete a given schedule of withdrawals. For withdrawals of a constant amount at equidistant times, necessary conditions are given on the initial investment and parameters of the wealth process such that $k$ withdrawals can be made with 95\% confidence. When the initial investment is in the S\&P Composite Index and $2\leq k\leq 16$, then the initial investment must be at least $k$ times the amount of each withdrawal.}






\maketitle

\section{Introduction}\label{sec1}
A regular investing schedule involves predetermined investment amounts and times. In \cite{oldDCA}, a log-Normal lower bound is given for the returns of a regular investing schedule, provided the asset wealth process is a geometric Brownian motion. Here, the lower bound is generalized to geometric L\'evy alpha-stable asset wealth processes. The main advantage of this generalization is that now, a lower bound can be produced when the asset wealth process has heavy tails and skewness. The generalized lower bound presented here has a log-L\'evy alpha-stable distribution, and it is given as a lower bound for terminal wealth, instead of returns. Note that return indicates terminal wealth divided by the total invested.

An interesting application of the lower bound on terminal wealth addresses the probability to make a sequence of withdrawals. A transformation of the recursion indicating how much money is left after each withdrawal allows the withdrawals process to be considered in the framework of a regular investing schedule. Section \ref{notation} describes the transformation in detail. Since the transformed withdrawals process is a regular investing schedule, the lower bound on terminal wealth can be applied. As a result, it is possible to establish a relationship between given parameters and the probability to make a sequence of withdrawals.

The generalized lower bound is applied to dollar cost averaging (DCA). DCA is a particular regular investing schedule, where a constant amount is invested at equidistant time steps. When a constant amount is withdrawn at equidistant time steps, the transformed withdrawals process is DCA. Thus application of the generalized lower bound to DCA addresses withdrawals as well. Some applications use historic price data from the S\&P Composite Index. Figure \ref{fig:genpic} illustrates a regular investing schedule vs a scedule of withdrawals. 

\begin{figure}[h] 
\setlength{\unitlength}{0.4cm}
\begin{picture}(20,18)
\thicklines
\put(4,17){Regular Investment}
\put(1,7){\begin{turn}{90}\text{time}\end{turn}}
\put(4,1){\vector(0,1){15}}
\put(4,1){\line(-1,0){.4}}
\put(2.8,.7){0}
\put(5,1){\vector(-1,0){1}}
\put(5,.8){{\footnotesize buy $\frac{c_0}{X(0)}$ shares}}
\put(4,3){\line(-1,0){.4}}
\put(2.8,2.7){$t_1$}
\put(5,3){\vector(-1,0){1}}
\put(5,2.8){{\footnotesize buy $\frac{c_1}{X(t_1)}$ shares}}
\put(4,8){\line(-1,0){.4}}
\put(2.8,7.7){$t_2$}
\put(5,8){\vector(-1,0){1}}
\put(5,7.8){{\footnotesize buy $\frac{c_2}{X(t_2)}$ shares}}
\put(4,11){\line(-1,0){.4}}
\put(2.8,10.7){$t_3$}
\put(5,11){\vector(-1,0){1}}
\put(5,10.8){{\footnotesize buy $\frac{c_3}{X(t_3)}$ shares}}
\put(4,13){\line(-1,0){.4}}
\put(2.8,12.7){$t_4$}
\put(5,13){\vector(-1,0){1}}
\put(5,12.8){{\footnotesize buy $\frac{c_4}{X(t_4)}$ shares}}

\put(18,17){Withdrawls}
\put(18,1){\vector(0,1){15}}
\put(18,1){\line(-1,0){.4}}
\put(16.8,.7){0}
\put(19,1){\vector(-1,0){1}}
\put(19,.8){{\footnotesize buy $\frac{P}{X(0)}$ shares}}
\put(18,3){\line(-1,0){.4}}
\put(16.8,2.7){$t_1$}
\put(19,3){\vector(-1,0){1}}
\put(19,2.8){{\footnotesize sell $\frac{w_1}{X(t_1)}$ shares}}
\put(18,8){\line(-1,0){.4}}
\put(16.8,7.7){$t_2$}
\put(19,8){\vector(-1,0){1}}
\put(19,7.8){{\footnotesize sell $\frac{w_2}{X(t_2)}$ shares}}
\put(18,11){\line(-1,0){.4}}
\put(16.8,10.7){$t_3$}
\put(19,11){\vector(-1,0){1}}
\put(19,10.8){{\footnotesize sell $\frac{w_3}{X(t_3)}$ shares}}
\put(18,13){\line(-1,0){.4}}
\put(16.8,12.7){$t_4$}
\put(19,13){\vector(-1,0){1}}
\put(19,12.8){{\footnotesize sell $\frac{w_4}{X(t_4)}$ shares}}
\end{picture}
\caption{Illustrates regular investment vs withdrawals for a particular stock using predetermined amounts and times. $X:[0,\infty)\to(0,\infty)$ is the share price, and the $c_k,\ w_k$ and $P$ are positive constants. The ability to buy fractional shares is assumed. Dollar cost averaging occurs when $t_k=k$ and $c_k=c_0$ for $k=0,1,...$.}
\label{fig:genpic}
\end{figure}
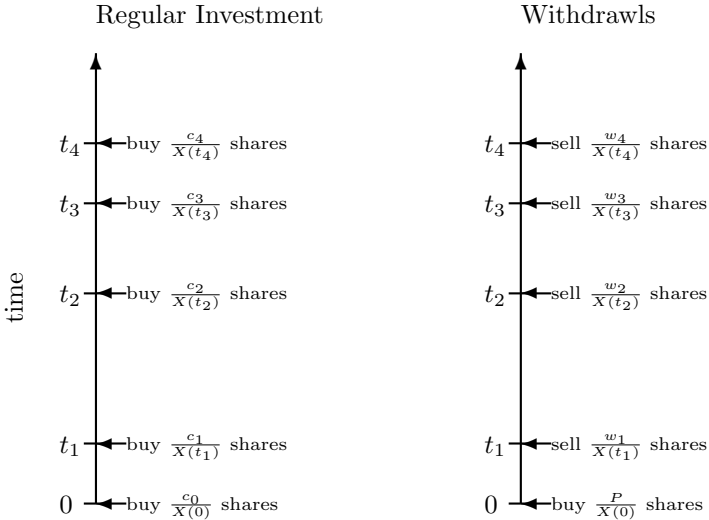

\subsection{Literature Review}
Practical justification for using a geometric L\'evy alpha-stable asset wealth process comes in several stages. First, values of the asset wealth process are collected at equidistant times. Then log-returns between successive values are computed. Last, log-returns are fit to a L\'evy alpha-stable distribution and checked for independence. Note that independence of log-returns follows from the random walk hypothesis. For this reason, many of the works mentioned in the next paragraph do not address independence explicitly.

L\'evy alpha-stable log-returns have been justified in several cases. To give an idea of what ranges of the shape parameter $\alpha$ are practical, values of $\alpha$ are also recorded. In \cite{nolan2003modeling}, the log difference of successive exchange rates is shown to fit a L\'evy alpha-stable random variable. In particular, daily British Pound vs. German Mark exchange rates are fitted over 1980 to 1996, and monthly US Dollar vs Tanzanian Shilling exchange rates are fitted over 1975 to 1997. The former estimates $\alpha=1.495$, and the latter estimates $\alpha=1.088$. US stock log-returns are shown to follow a L\'evy alpha-stable distribution in \cite{fama1965behavior,leitch1975estimation}. Chinese stock daily log-returns are shown to follow a L\'evy alpha-stable distribution in \cite{xu2011modeling}. For US stocks, $\alpha$ is generally estimated in the range 1.6 to 2. For Chinese stocks, $\alpha$ is estimated around 1.4. In \cite{cornew1984stable}, daily log-returns from commodity futures are well-fit to a L\'evy alpha-stable distribution, and $\alpha$ is typically estimated to be in the range 1.5 to 2. Overall, $\alpha>1$ is practical.

Given a regular investing schedule and a random asset wealth process, terminal wealth is a sum of products of random variables. Since elements of the sum are dependent, the distribution of terminal wealth is complicated. As a result, it is desirable to find a bound on terminal wealth that has a simple distribution. Then it is possible to make interesting statements about terminal wealth using the simple distribution. 

When the asset wealth process is a geometric Brownian motion, terminal wealth of a regular investing schedule is a sum of dependent log-Normal random variables. The sum can be estimated as a log-Normal random variable, as in \cite{mehta2007approximating,schwartz1982distribution}, but these estimates are not necessarily lower bounds. A lower bound for the sum is given in \cite{dhaene2002concept,dhaene2002conceptT}, but it is a lower bound in the convex order sense. A lower bound in probability is given in \cite{oldDCA}, and that lower bound is generalized here to geometric L\'evy alpha-stable wealth processes.

Terminal wealth of a continuous investment schedule has been considered \cite{milevsky2003continuous}. In particular, the investment schedule is the continuous version of DCA, and the asset wealth process is assumed to be a geometric Brownian motion w.r.t. time. As a result, the terminal wealth is integrated geometric Brownian motion.

Much of the recent research considering a schedule of withdrawals focuses on variable annuities. In a simple example of a variable annuity, an individual first pays a premium to an insurance company, and the insurance company invests the premium into a mutual fund. Then, the insurance company makes payouts to the individual at standard intervals (e.g. monthly, quarterly, semiannually or annually). The size of each payout depends on the options and benefits specific to the variable annuity contract. Furthermore, each set of options and benefits has associated fees and penalties written into the contract. One benefit that has recieved substantial attention is the guaranteed minimum withdrawal benefit, where an individual is guaranteed to at least withdraw the fee- and penatly-adjusted premium. Pricing of variable annuities with a guaranteed minimum withdrawal benefit has been addressed in several cases (e.g. \cite{milevsky2006financial,dai2008guaranteed}). Here, results considering withdrawls use a single initial investment, or premium, and then a sequence of withdrawls. 

Empirical research addresses the sustainability of a constant withdrawal amount at equidistant times, given the initial investment. In \cite{cooley1998retirement}, annual withdrawals are made from portfolios consisting of US stocks and bonds, using data from 1926 to 1995. The withdrawals are made over the periods 15, 20, 25 and 30 years. Portfolios having only stocks are better able to sustain higher withdrawal amounts, like 12\% of the initial investment. Portfolios having 50\% stocks and 50\% bonds are better able to withstain withdrawal amounts that are 7\% of the initial investment. In \cite{cooley2003does}, the effect of international diversification on sustainability of withrawals from a stock and bond portfolio is shown to be minor when the portfolio has at least 50\% stocks; the time period considered is 1970 to 2001. Here, applications consider the probability to make a sequence of withdrawals, given an initial investment in just US stocks.


\subsection{Main Results}
Investment in exactly one asset is considered. All results require the assumption that the logarithm of the asset wealth process is a L\'evy alpha-stable process. 

For a regular investment schedule, a lower bound is given on the terminal wealth such that the logarithm of the lower bound has a L\'evy alpha-stable distribution (see Theorems \ref{tG}, \ref{t1} and \ref{tcont}). Theorem \ref{tG} provides recursive expression of the lower bound for the general investing schedule. Theorems \ref{t1} and \ref{tcont} provide closed form expression of the lower bound for DCA and continuous DCA, respectively. 

For a regular withdrawal schedule, an upper bound is given on the probability of being able to successfully execute all planned withdrawals. For the general case, where withdrawal amounts and times are unconstrained, the upper bound is a result of Theorems \ref{tG*} and \ref{tW}. For the case where withdrawal amounts and time intervals are equal, the upper bound is a result of Theorems \ref{t1*} and \ref{tW}. Remark \ref{rcont} describes the upper bound for the case when withdrawals are continuous.

\subsection{Applications}
Applications are split into general applications and applications with data. General applications consider a range of parameters for the L\'evy alpha-stable asset wealth process. Applications with data focus only on parameters that provide a good fit to the data. Data is from the S\&P Composite Index, specifically the annual data from 1871 to 2020, including reinvested dividends and adjustment for inflation. This S\&P index tracks the weighted average of stock prices for the largest US companies, where each company's weight is equal to its market capitalization. Funds tracking this index, like the S\&P 500, are very popular among investors. Data is taken from Robert Shiller's online data library, and further details on this data can be found in Section \ref{sec51}. 

General applications address the lump sum discount, discrete withdrawals and continuous withdrawals. The lump sum discount was introduced in \cite{oldDCA}. It indicates a lump sum investment that has a terminal wealth distribution that is no better than a given regular investing schedule. Note that a lump sum investment is a single investment at time 0. Here, the lump sum discount is described for continuous DCA using Theorem \ref{tcont}. The lump sum discount is invariant to the skewness and scale parameters of the asset wealth process, so it is illustrated for various location and shape parameters. In general, investing a similar amount in lump sum for less than half the time is no better than using continuous DCA.

The application addressing discrete withdrawals considers equal withdrawals at equidistant times. Suppose an investor wishes to make $k$ withdrawals with a given level of confidence, shape parameter and skewness parameter. Using Theorems \ref{t1*} and \ref{tW}, a procedure is outlined to give conditions on the location and scale parameters that force the initial investment to be at least the total intended to be withdrawn. The procedure is demonstrated for 95\% confidence and particular shape and skewness parameters. 

The application addressing continuous withdrawals suppose an investor wishes to continuously withdraw a specified amount, at a constant rate, with a given level of confidence. For 95\% confidence with various shape, skewness and location parameters, conditions are given on the scale parameter that force the initial investment to be at least the total intended to be withdrawn.  

Applications with S\&P data first address the fit of log-returns to a L\'evy alpha-stable process. In \cite{oldDCA}, log-returns are fit to a Brownian motion, which is a particular L\'evy alpha-stable process. Main results allow log-returns to be fit to L\'evy alpha-stable processes that are not Brownian motions. Consequently, a L\'evy alpha-stable process is found that fits log-returns better than the Brownian motion of \cite{oldDCA}. The lower bound on returns from Theorem \ref{t1} is compared using the new L\'evy alpha-stable process fit vs the Brownian motion fit. Quantiles are practically identical, except the quantiles less than .05. The quantiles under .05 are less with the L\'evy alpha-stable process, implying there is additional downside risk that the Brownian motion fit misses. However, if a decision is made based on quantiles no less than .05, then the decision will likely be invariant to the L\'evy alpha-stable process fit vs the Brownian motion fit.

Next, applications with S\&P data consider discrete withdrawals. In order to make $k\in\{2,3,...,16\}$ equal, annual withdrawals from the S\&P Composite Index, with 95\% confidence, it is necessary for the initial investment to be at least $k$ times the amount of each withdrawal. 

Last, continuous withdrawals are considered. Given various levels of confidence, necessary initial investments are given to continuously withdraw a particular amount from the S\&P Composite Index over $n\in\{2,6,12,20,30,42\}$ years. To continuously withdraw $d\$$ with a high level of confidence, like 99\%, the initial investment needs to be at least $d\$$, when continuous withdrawals are made over $n\leq 42$ years. For lower confidence, like 60\%, the initial investment needs to be at least $d/3\$$.

\subsection{Organization}
Section \ref{sec2} provides the problem setup and main results. Main results are split into section \ref{mainPI} for regular investment and section \ref{mainW} for withdrawals. Sections \ref{genApp} and \ref{dataApp} provide application of the main results given in section \ref{sec2}. Applications are split into section \ref{genApp} for general applications and section \ref{dataApp} for applications using data. Section \ref{sec4} provides closing remarks, including a discussion of related future research ideas. Appendix \ref{secA1} provides proofs of the theorems stated in Section \ref{sec2}.

\section{Definitions \& Main Results}\label{sec2}
\subsection{Notation}\label{notation}
Let $\log X:[0,\infty)\times\Omega\to(0,\infty)$ be a L\'evy alpha-stable process defined on the probability space $(\Omega,\mathcal{F},\mathbb{P})$. When using $X(t,\omega)$ as a function of $\omega$ only, write $X(t)$ in place of $X(t,\omega)$. For each $t\in[0,\infty)$, the characteristic function of $\log\frac{X(t)}{X(s)}$, $t>s$, is given by
\begin{equation*}
\Phi(\theta)=\begin{cases}
\exp\{i\theta (t-s)\mu-(t-s)\lvert \sigma\theta\rvert^{\alpha}(1-i\beta\frac{\theta}{\lvert\theta\rvert}\tan\frac{\pi\alpha}{2})\}&\alpha\neq1\\
\exp\{i\theta (t-s)\mu-(t-s)\lvert \sigma\theta\rvert(1+i\beta\frac{2\theta}{\pi\lvert\theta\rvert}\log\lvert\theta\rvert\}&\alpha=1
\end{cases},
\end{equation*}
where $\alpha\in(0,2]$, $\beta\in[-1,1]$, $\sigma\in(0,\infty)$ and $\mu\in\mathbb{R}$. 

For arbitrary $Z$, use $Z\sim\mathcal{S}(\alpha^*,\beta^*,\sigma^*,\mu^*)$ to indicate that $Z$ is a L\'evy alpha-stable random variable with shape $\alpha^*$, skewness $\beta^*$, scale $\sigma^*$ and location $\mu^*$. Moreover, $Z$ has characteristic function
\begin{equation*}
\Phi(\theta)=\begin{cases}
\exp\{i\theta\mu^*-\lvert \sigma^*\theta\rvert^{\alpha^*}(1-i\beta^*\frac{\theta}{\lvert\theta\rvert}\tan\frac{\pi\alpha^*}{2})\}&\alpha^*\neq1\\
\exp\{i\theta\mu^*-\lvert \sigma^*\theta\rvert(1+i\beta^*\frac{2\theta}{\pi\lvert\theta\rvert}\log\lvert\theta\rvert\}&\alpha^*=1
\end{cases}.
\end{equation*}
Say that $(\sigma^*)^{-1}\mu^*$ is the Sharpe ratio of $Z$.

Let $\{t_k\}_{k=0}$ be a sequence in $[0,\infty)$ with $t_0=0$. Suppose that $c_k$ is invested at each time step $t_k$. Let $X_k=\frac{X(t_k)}{X(t_{k-1})}$ for $k=1,2,...$. The returns at time step $t_k$ are given by $R_k=(\sum_{j=0}^{k-1}c_j)^{-1}Y_k$ for $k=1,2,...$, where the $Y_k$ are computed recursively via
\begin{equation}
\begin{split}
Y_1&=c_0X_1,\\
Y_k&=X_k\cdot(Y_{k-1}+c_{k-1}),\quad k=2,3,...
\end{split}
\label{recursion}
\end{equation}
Note that $Y_k$ is the terminal wealth at time $t_k$. 

Next consider the analogue to \eqref{recursion}, but for withdrawals. In particular, suppose $P$ is invested at time $t_0$, and $w_k$ is withdrawn at each time step $t_k$, where $k=1,2,...$. Then the amount remaining at time $t_k$ after withdrawal is computed recursively via
\begin{equation}
\begin{split}
W_1&=PX_1-w_1,\\
W_k&=W_{k-1}X_k-w_k,\quad k=2,3,...
\end{split}
\label{recursionW0}
\end{equation}
Note that $w_1,...,w_k$ can be withdrawn at times $t_1,...,t_k$ if and only if $W_k\geq0$. Observe that
\begin{equation}
\begin{split}
P&=X_1^{-1}\cdot(W_1+w_1),\\
W_{k-1}&=X_k^{-1}\cdot(W_k+w_k),\quad k=2,3,...
\end{split}
\label{recursionW}
\end{equation}

The following abbreviations will be used to shorten descriptions: dollar cost averaging (DCA), lump sum (LS), upper bound (UB) and lower bound (LB). In all expressions, $\log$ indicates the natural logarithm. When convenient, expressions like $X=Y$ w.p.1 will be abbreviated with $X=Y$. Equality in distribution will be denoted with $X\stackrel{d}{=}Y$.

\subsection{Main Results for Regular Investment}\label{mainPI}
Definition \ref{defz} recursively constructs the lower bound of terminal wealth. Theorem \ref{tG} describes the lower bound of terminal wealth recursively. Under the conditions $c_k=1$ and $t_k=k$ for $k=0,1,...$, Theorem \ref{t1} describes the lower bound of DCA terminal wealth in closed form. Theorem \ref{tcont} provides an upper bound on the cumulative distribution function of terminal wealth when $c_k=\frac{1}{n}$ and $t_k=\frac{k}{n}$ for $k=0,1,...,n-1$, and the limit is taken as $n\to\infty$. This is the continuous version of DCA.

Note that a lower bound given for terminal wealth is easily transformed into a lower bound for returns. In particular, if $Z_k\leq Y_k$ w.p.1 and $\log Z_k\sim\mathcal{S}(\alpha,\beta,\sigma_k,\mu_k)$, then $(\sum_{j=0}^{k-1}c_j)^{-1}Z_k\leq R_k$ w.p.1 and \begin{equation*}
\log \Big((\sum_{j=0}^{k-1}c_j)^{-1}Z_k\Big)\sim\mathcal{S}(\alpha,\beta,\sigma_k,\mu_k-\log\sum_{j=0}^{k-1}c_j).
\end{equation*}

\begin{definition}
Set $Z_1=Y_1$ and 
\begin{equation*}
Z_k=X_k(a_{k-1}+c_{k-1})\Bigg(\frac{Z_{k-1}}{a_{k-1}}\Bigg)^{b_{k-1}},\quad k=2,3,...,
\end{equation*}
where $a_{k-1}>0$ and $b_{k-1}=\frac{a_{k-1}}{a_{k-1}+c_{k-1}}$.
\label{defz}
\end{definition}

\begin{remark}
$R_k$ and $(\sum_{j=0}^{k-1}c_j)^{-1}Z_k$ are invariant to a rescaling of the investment amounts. In particular, if the substitution $c_j\leftarrow \lambda c_j$ is made for each $j$, then $R_k$ and $(\sum_{j=0}^{k-1}c_j)^{-1}Z_k$ do not change.
\end{remark}

\begin{theorem}
For every $k=1,2,...$, $Z_k\leq Y_k$ w.p.1 and $\log Z_k\sim\mathcal{S}(\alpha,\beta,\sigma_k,\mu_k)$, where parameters are given recursively via 
\begin{equation*}
\begin{split}
\sigma_1^\alpha&=(c_0\sigma)^{\alpha}t_1,\quad\mu_1=\begin{cases}
c_0\mu t_1&\alpha\neq1\\
c_0\mu t_1-\frac{2}{\pi}\beta c_0\sigma\sqrt[\alpha]{t_1}\log c_0&\alpha= 1
\end{cases},\\
m_k&=b_{k-1}(\mu_{k-1}-\log a_{k-1})+\log(a_{k-1}+c_{k-1})+\mu(t_k-t_{k-1}),\\
\mu_k&=\begin{cases}
m_k&\alpha\neq1\\
m_k-\frac{2}{\pi}\beta b_{k-1}\sigma_{k-1}\log b_{k-1}&\alpha=1
\end{cases},\\
\sigma_{k}^{\alpha}&=(b_{k-1}\sigma_{k-1})^{\alpha}+\sigma^{\alpha}(t_k-t_{k-1}).
\end{split}
\end{equation*}
\label{tG}
\end{theorem}

\begin{remark}
In general $a_{k-1}$ can be any positive real number. A good choice for $a_{k-1}$ that leads to nice simplification is $a_{k-1}=\exp\mu_{k-1}$. Then the term $b_{k-1}(\mu_{k-1}-\log a_{k-1})$ in Theorem \ref{tG} disappears. Moreover, $a_{k-1}=\exp\mu_{k-1}$ maximizes $m_k$ over $a_{k-1}$.
\end{remark}

\begin{theorem}
Let $\alpha\neq1$ and $\mu\neq0$. Suppose $a_{k+1}=\exp\mu_{k+1}$, $c_k=1$ and $t_k=k$ for $k=0,1,...$. Then for every $k=1,2,...$, $Z_k\leq Y_k$ w.p.1 and $\log Z_k\sim\mathcal{S}(\alpha,\beta,\sigma_k,\mu_k)$, where 
\begin{equation*}
\begin{split}
\mu_k&=\mu+\log\frac{\exp(\mu k)-1}{\exp\mu-1}\\
\sigma_k^{\alpha}&=\sigma^{\alpha}\Bigg(1+\sum_{j=1}^{k-1}\Big(1-\frac{\exp(\mu j)-1}{\exp(\mu k)-1}\Big)^{\alpha}\Bigg).
\end{split}
\end{equation*}
\label{t1}
\end{theorem}

\begin{remark}
In the construction of Theorem \ref{t1}, $R_k=\frac{1}{k}Y_k$. Furthermore, $\sigma_k$ is increasing with $k$. If $\mu>0$, then $\mu_k$ is increasing with $k$. If $\mu<0$, then $\mu_k$ is decreasing with $k$. 
\end{remark}

\begin{theorem}
Let $\alpha\neq1$ and $\mu\neq0$. Let $\mathcal{Y}_n$ denote the returns at $t=1$ given $a_{k+1}=\exp\mu_{k+1}$, $c_k=\frac{1}{n}$ and $t_k=\frac{k}{n}$ for $k=0,1,...,n-1$. Then $\mathbb{P}(Z\leq x)\geq\lim_{n\to\infty}\mathbb{P}(\mathcal{Y}_n\leq x)$ for all $x\in\mathbb{R}$, where $\log Z\sim\mathcal{S}(\alpha,\beta,\sigma^*,\mu^*)$ and 
\begin{equation*}
\begin{split}
\mu^*&=\log\frac{\exp\mu-1}{\mu}\\
(\sigma^*)^{\alpha}&=\sigma^{\alpha}\int_0^1\Big(\frac{1-\exp(-\mu x)}{1-\exp(-\mu)}\Big)^{\alpha}dx.
\end{split}
\end{equation*}
\label{tcont}
\end{theorem}

\begin{remark}
In the construction of Theorem \ref{tcont}, $\mathcal{Y}_n$ is the terminal wealth and return at time 1, since the total invested is 1. Furthermore, the integral needed to compute $\sigma^*$ can be bounded, and it is well-approximated when $\mu$ is close to 0. Observe that for $x\in[0,1]$,
\begin{equation}
\big(1-\exp(-\mu)\big)x\leq 1-\exp(-\mu x)\leq \big(1-\exp(-\mu)\big)x^r,\quad r=\frac{\mu}{\exp\mu-1}.
\label{xxr}
\end{equation}
Proof of the upper bound in \eqref{xxr} is given in Appendix \ref{secA1}. It follows that
\begin{equation*}
\frac{\chi(\mu)}{\alpha+1}\leq \chi(\mu)\int_0^1\Big(\frac{1-\exp(-\mu x)}{1-\exp(-\mu)}\Big)^{\alpha}dx\leq \frac{\chi(\mu)}{r\alpha+1},
\end{equation*}
where $\chi(\mu)=1$ if $\mu\geq0$ and $\chi(\mu)=-1$ otherwise. 
\end{remark}

\subsection{Main Results for Withdrawls}\label{mainW}
For the remainder of this subsection, fix $k\geq1$. In \eqref{recursionW}, $P$ is deterministic and $W_k$ is random. Alternatively, consider the recursion
\begin{equation}
\begin{split}
P^*&=X_1^{-1}\cdot(W_1^*+w_1),\\
W_{j-1}^*&=X_j^{-1}\cdot(W_j^*+w_j),\quad j=2,3,...,k-1,\\
W_{k-1}^*&=w_kX_k^{-1},
\end{split}
\label{recursionW2}
\end{equation}
where $P^*$ is random. Then for each $\omega\in\Omega$, $P^*(\omega)>P\iff W_k(\omega)<0$. Note that the if and only if holds because $W_k$ is a strictly increasing function of $P$. Theorem \ref{tG*} provides a lower bound for $P^*$, where parameters are given recursively. Theorem \ref{t1*} provides a lower bound for $P^*$ when every withdrawal is the same amount, and withdrawals are made at equidistant times. Furthermore, Theorem \ref{t1*} expresses parameters of the lower bound in closed form. Theorem \ref{tW} uses the lower bound of $P^*$ to construct an upper bound for $\mathbb{P}(W_k\geq0)$. Recall that $W_k\geq0$ if and only if withdrawals $w_1,...,w_k$ can be made. Remark \ref{rcont} discusses continuous withdrawals.

\begin{definition}
Set $Z_{k-1}^*=W_{k-1}^*$ and 
\begin{equation*}
Z_{j-1}^*=X_j^{-1}(a_j+w_j)\Bigg(\frac{Z_j^*}{a_j}\Bigg)^{b_j},\quad j=k-1,...,1,
\end{equation*}
where $a_j>0$ and $b_j=\frac{a_j}{a_j+w_j}$.
\label{defz*}
\end{definition}

\begin{theorem}
$Z_0^*\leq P^*$ w.p.1 and $\log Z_0^*\sim\mathcal{S}(\alpha,-\beta,\sigma_0,\mu_0)$, where parameters are given recursively via 
\begin{equation*}
\begin{split}
\sigma_{k-1}^\alpha&=(w_k\sigma)^{\alpha}(t_k-t_{k-1}),\\
\mu_{k-1}&=\begin{cases}
-w_k\mu(t_k-t_{k-1})&\alpha\neq1\\
-w_k\mu(t_k-t_{k-1})+\frac{2}{\pi}\beta w_k\sigma\sqrt[\alpha]{t_k-t_{k-1}}\log w_k&\alpha= 1
\end{cases},\\
m_{j-1}&=b_j(\mu_j-\log a_j)+\log(a_j+w_j)-\mu(t_j-t_{j-1}),\\
\mu_{j-1}&=\begin{cases}
m_{j-1}&\alpha\neq1\\
m_{j-1}+\frac{2}{\pi}\beta b_j\sigma_j\log b_j&\alpha=1
\end{cases},\\
\sigma_{j-1}^{\alpha}&=(b_j\sigma_j)^{\alpha}+\sigma^{\alpha}(t_j-t_{j-1}).
\end{split}
\end{equation*}
\label{tG*}
\end{theorem}

\begin{remark}
Like in Section \ref{mainPI}, $a_j$ can be any positive real number, and a good choice for $a_j$ is $a_j=\exp\mu_j$. 
\end{remark}

\begin{theorem}
Let $\alpha\neq1$ and $\mu\neq0$. Suppose $a_j=\exp\mu_j$ for $j=1,...,k-1$, and $w_j=1$, $t_j=j$ for $j=1,...,k$. Then $Z_0^*\leq P^*$ w.p.1 and $\log Z_0^*\sim\mathcal{S}(\alpha,-\beta,\sigma_0,\mu_0)$, where 
\begin{equation*}
\begin{split}
\mu_0&=-\mu+\log\frac{\exp(-\mu k)-1}{\exp(-\mu)-1}\\
\sigma_0^{\alpha}&=\sigma^{\alpha}\Bigg(1+\sum_{j=1}^{k-1}\Big(1-\frac{\exp(-\mu j)-1}{\exp(-\mu k)-1}\Big)^{\alpha}\Bigg).
\end{split}
\end{equation*}
\label{t1*}
\end{theorem}

\begin{theorem}
The probability of being able to make withdrawals $w_1,...,w_k$ is no greater than the probability that $Z_0^*\leq P$, i.e. $\mathbb{P}(W_k\geq0)\leq\mathbb{P}(Z_0^*\leq P)$.
\label{tW}
\end{theorem}

\begin{remark}
\label{rcont}
Let $\alpha\neq1$ and $\mu\neq0$. Suppose one wishes to make continuous withdrawals at a constant rate from time 0 to time 1, such that a total of exactly 1 unit has been withdrawn by time 1. Given the initial investment of $P$ at time 0, the probability of success is no greater than $\mathbb{P}(Z^*\leq P)$, where $\log Z^*\sim\mathcal{S}(\alpha,-\beta,\sigma^*,\mu^*)$ and $\sigma^*,\ \mu^*$ are computed as in Theorem \ref{tcont}, except with the substitution $\mu\leftarrow -\mu$. This result is just a combination of Theorems \ref{tcont} and \ref{tW}, using \eqref{recursionW2}.
\end{remark}

\section{General Applications}\label{genApp}
\subsection{Lump sum discount}
The lump sum discount, introduced in \cite{oldDCA}, can be evaluated using Theorems \ref{tG}, \ref{t1} or \ref{tcont}. Here, it is evaluated for Theorem \ref{tcont}. The idea behind the lump sum discount is to identify a lump sum investment having a terminal wealth that is no better (in distribution) than the terminal wealth of continuous DCA. 

The goal is to find a lump sum investment $x$ that, after $s$ time units, matches the terminal wealth distribution of the lower bound for continuous DCA. In particular, require
\begin{equation*}
x\cdot\frac{X(s)}{X(0)}\stackrel{d}{=}Z,\quad \log Z\sim\mathcal{S}(\alpha,\beta,\sigma^*,\mu^*),
\end{equation*}
where $\sigma^*$ and $\mu^*$ are as in Theorem \ref{tcont}. It follows that 
\begin{equation}
\begin{split}
s&=\int_0^1\Big(\frac{1-\exp(-\mu x)}{1-\exp(-\mu)}\Big)^{\alpha}dx,\\
x&=\frac{\exp\mu-1}{\mu\exp(\mu s)}.
\end{split}
\label{xsdef}
\end{equation}
Since the continuous DCA of Theorem \ref{tcont} invests a total of 1 currency unit over 1 unit of time, the lump sum discount is $(x,s)$, indicating a lump sum investment of $x$ currency units for $s$ time units has no better terminal wealth than the continuous DCA. The result is also scalable, meaning if $d\$$ is invested over $t$ time units using continuous DCA, then a lump sum investment of $xd\$$ for $st$ time units has a terminal wealth that is no better than the continuous DCA. Figure \ref{fig:xs} shows $x$ and $s$ for various $\mu$ and $\alpha$. In general, investing a similar amount in lump sum for less than half the time is no better than using continuous DCA. 
\begin{figure}[H]
  \includegraphics[width=\linewidth]{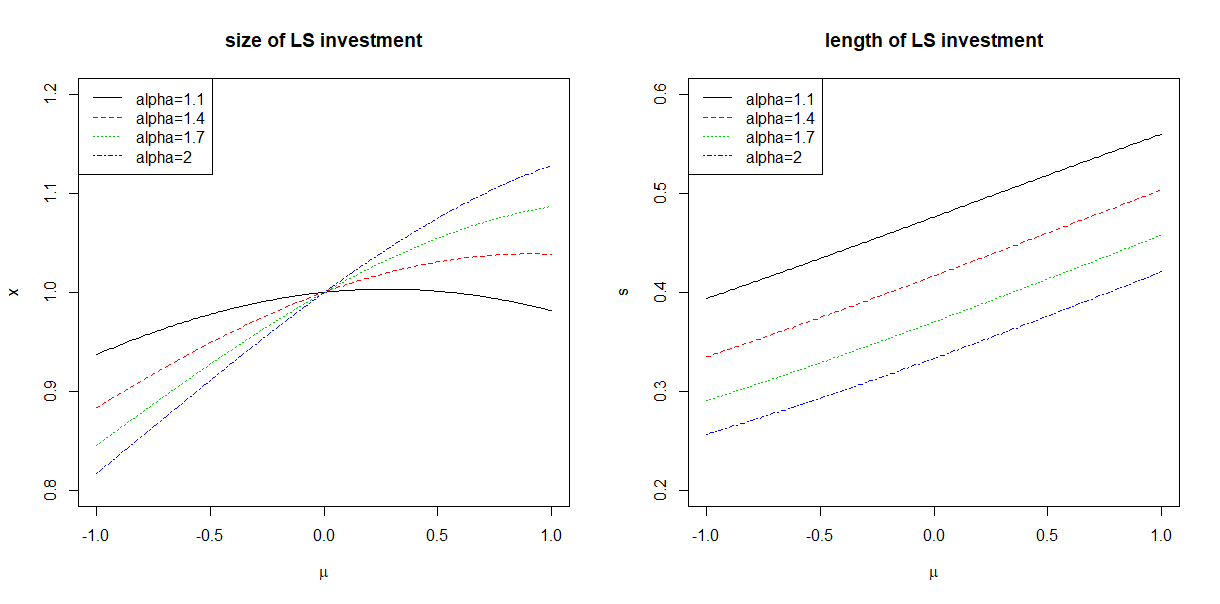}
  \caption{Illustrates $x$ and $s$ from the lump sum discount. $x$ and $s$ were computed using \eqref{xsdef}. Romberg integration was used to compute $s$.}
  \label{fig:xs}
\end{figure}

\subsection{Discrete Withdrawls}\label{genW}
Suppose an investor wishes to make $k$ equal withdrawals at equidistant time steps, with a given level of confidence. Theorems \ref{t1*} and \ref{tW} make it possible to specify a necessary initial investment $P$, meaning the investor must invest at least $P$ at time 0 in order to achieve the given level of confidence. For the remainder of this subsection, use $P$ to denote that necessary initial investment.

The procedure is as follows. Let $C\in(0,1)$ denote the given level of confidence. Set $C=\mathbb{P}(Z_0^*\leq P)$, and solve for $P$, where $Z_0^*$ is as in Theorem \ref{t1*}. By Theorem \ref{tW}, the probability of being able to make the $k$ withdrawals is no greater than $C$. If an initial investment $Q$ is made instead of $P$, and $Q<P$, then $\mathbb{P}(Z_0^*\leq Q)<C$. It follows from Theorem \ref{tW} that the probability of being able to make the $k$ withdrawals is less than $C$. Therefore, $Q$ is not a sufficient initial investment. 

Without loss of generality, suppose the investor makes withdrawals as in Theorem \ref{t1*}. Then 1 unit is withdrawn at each time $1,2,...,k$. An interesting result to pursue is: Given $C$, $k$, $\alpha$ and $\beta$, provide sufficient $\mu$ and $\sigma$ such that $P\geq k$, meaning the initial investment must be at least the total amount intended to be withdrawn. This result is achieved for $C=.95$, $k\in\{2,3,...,60\}$, $\alpha\in\{1+.01n:\ n=1,...,100\}$ and $\beta=0$. Since $C$ and $\beta$ are fixed, sufficient ratios $\mu/\sigma$ are given as a function of $k$ and $\alpha$ in \eqref{mskalpha}. Note that $\mu\in\{.001n:\ n=1,...,500\}$ is required for the result to hold. The $\Lambda_{ij}$ are provided in \eqref{beta}, where the matrix $\Lambda$ stores $\Lambda_{ij}$ in its $i$th row and $j$th column. The interpretation of the result is as follows: If $\mu\in\{.001n:\ n=1,...,500\}$, $k\in\{2,3,...,60\}$, $\alpha\in\{1+.01n:\ n=1,...,100\}$, $\beta=0$ and \eqref{mskalpha} holds, then in order to make $k$ equal withdrawals at each time $1,2,...,k$, with 95\% confidence, it is necessary for the initial investment to be at least $k$ times the amount of each withdrawal.

\begin{equation}
\begin{split}
\frac{\mu}{\sigma}&\leq f_0(\alpha)+f_1(\alpha)\sqrt{k}+f_2(\alpha)\sqrt[4]{k}+f_3(\alpha)\sqrt[8]{k},\\
f_i(\alpha)&=\sum_{j=0}^5\Lambda_{ij}\alpha^j.
\end{split}
\label{mskalpha}
\end{equation}

\footnotesize
\begin{equation}
\begin{split}
\Lambda&=\begin{bmatrix}
-2484.4421&	7053.6469&	-7910.3438&	4425.8516&	-1236.8914&	138.1989\\
203.2079&	-551.2414&	597.5744&	-324.7764&	88.4613&	-9.6564\\
-2681.4779&	7442.1643&	-8232.7022&	4561.2428&	-1265.6221&	140.6711\\
5061.6465&	-14204.7730&	15842.4735&	-8838.1550&	2467.1713&	-275.6974
\end{bmatrix}
\end{split}
\label{beta}
\end{equation}
\normalsize

The bound \eqref{mskalpha} was constructed via computer using Algorithm \ref{alg1}. First, $C$ and $\beta$ were fixed at $C=.95$ and $\beta=0$. Then, for each $\alpha\in\{1+.01n:\ n=1,...,100\}$, the returned $s_k$ were fit to the polynomial 
\begin{equation*}
\hat{f}_0(\alpha)+\hat{f}_1(\alpha)\sqrt{k}+\hat{f}_2(\alpha)\sqrt[4]{k}+\hat{f}_3(\alpha)\sqrt[8]{k}.
\end{equation*}
Last, each $\hat{f}_i(\alpha)$ was fit to the polynomial 
\begin{equation*}
\Lambda_{i0}+\Lambda_{i1}\alpha+\Lambda_{i2}\alpha^2+\Lambda_{i3}\alpha^3+\Lambda_{i4}\alpha^4+\Lambda_{i5}\alpha^5,
\end{equation*}
so that \eqref{mskalpha} implies $P\geq k$. Figure \ref{fig:sharpe} illustrates the upper bound \eqref{mskalpha} for different $\alpha$. This procedure can also be executed for different $C$ and $\beta$.

\begin{algorithm}
\caption{Build $\mu/\sigma$ upper bound}
\label{alg1}
\begin{algorithmic}
\Require $C\in(0,1)$, $\beta\in[-1,1]$, $Z\sim\mathcal{S}(\alpha,-\beta,1,0)$
\Require $\alpha\in\{1+.01n:\ n=1,...,100\}$
\State $M\gets\{.001n:\ n=1,...,500\}$ \Comment{set of $\mu$}
\State $K\gets\{2,3...,60\}$ \Comment{set of $k$}
\State $\mathbb{P}(Z\leq q)\gets C$  \Comment{solve for $q$}
\For{$k\in K$}
\State $s_k\gets 20$ \Comment{initial ratio $\frac{\mu}{\sigma}$}
\While{$\exists(\mu,\sigma)\in \{(\mu,\frac{\mu}{s_k}):\ \mu\in M\}$ s.t. $\exp(\mu_0+\sigma_0q)<k$} 
\State $s_k\gets s_k-.01$ \Comment{see Theorem \ref{t1*} for $\mu_0,\sigma_0$}
\EndWhile
\EndFor\\
\Return{$s_k,\ k\in K$}
\end{algorithmic}
\end{algorithm}

\begin{figure}[H]
  \includegraphics[width=\linewidth]{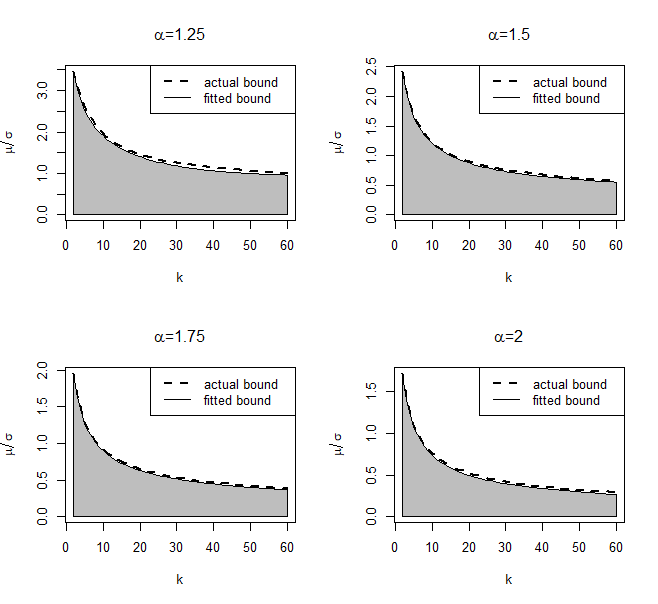}
  \caption{The line and shaded area illustrate \eqref{mskalpha}. The dashed line is the returned $s_k$ from Algorithm \ref{alg1} when $C=.95$ and $\beta=0$.}
  \label{fig:sharpe}
\end{figure}

\subsection{Continuous Withdrawls}\label{conW}
Suppose an investor wishes to continuously withdraw a total of 1 currency unit from time 0 to time 1, at a constant rate, with a given level of confidence. Remark \ref{rcont} makes it possible to specify a necessary initial investment $P$, meaning the investor must invest at least $P$ to achieve the given level of confidence. For the remainder of this subsection, use $P$ to denote that necessary initial investment.

The procedure is as follows. Let $C\in(0,1)$ denote the given level of confidence. Set $C=\mathbb{P}(Z^*\leq P)$, and solve for $P$, where $Z^*$ is as in Remark \ref{rcont}. The probability of being able to withdraw the 1 unit is no greater than $C$. Any initial investment less than $P$ results in a confidence level less than $C$. 
\begin{figure}[H]
  \includegraphics[width=\linewidth]{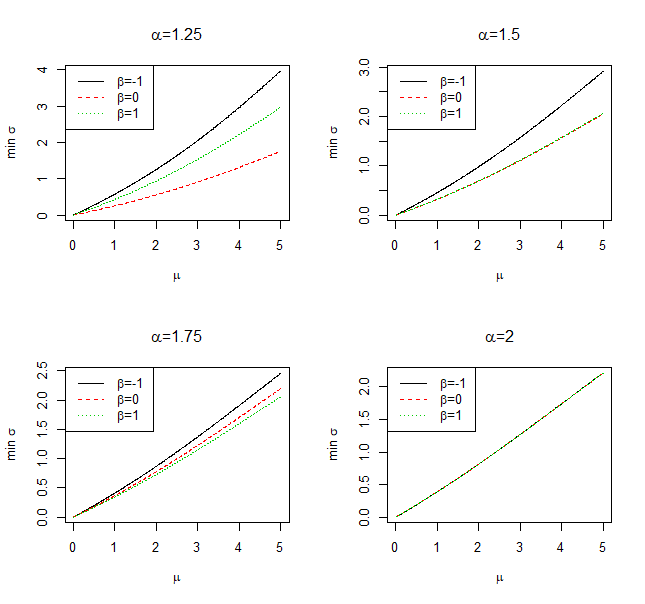}
  \caption{Illustrates $\min\sigma=\mu/s_{\mu}$ as a function of $\mu$, where $s_{\mu}$ is returned from Algorithm \ref{alg2}. $C$ is fixed at .95, while $\alpha$ and $\beta$ vary.}
  \label{fig:minsigma}
\end{figure}

An interesting result to pursue is: Given $C$, $\alpha$, $\beta$ and $\mu$, provide sufficient $\sigma$ such that $P\geq 1$, meaning the initial investment must be at least the total amount intended to be withdrawn. Algorithm \ref{alg2} computes $s_{\mu}$ for each $\mu\in\{.01n:\ n=1,...,500\}$, which is an upper bound on $\mu/\sigma$. It follows that $\mu/s_{\mu}$ is a lower bound on $\sigma$. Figure \ref{fig:minsigma} illustrates the lower bound on $\sigma$ for $C=.95$ and various $\alpha$ and $\beta$.

\begin{algorithm}
\caption{Build $\mu/\sigma$ upper bound}
\label{alg2}
\begin{algorithmic}
\Require $C\in(0,1)$, $\alpha\in(0,1]$, $\beta\in[-1,1]$, $Z\sim\mathcal{S}(\alpha,-\beta,1,0)$
\State $M\gets\{.01n:\ n=1,...,500\}$ \Comment{set of $\mu$}
\State $\mathbb{P}(Z\leq q)\gets C$  \Comment{solve for $q$}
\For{$\mu\in M$}
\State $s_{\mu}\gets 20$ \Comment{initial ratio $\frac{\mu}{\sigma}$}
\While{$\exp(\mu^*+\sigma^*q)<1$, where $(\mu,\sigma)=(\mu,\frac{\mu}{s_{\mu}})$} 
\State $s_{\mu}\gets s_{\mu}-.01$\Comment{see Remark \ref{rcont} for $\mu^*,\sigma^*$}
\EndWhile
\EndFor\\
\Return{$s_{\mu},\ \mu\in M$}
\end{algorithmic}
\end{algorithm}

\section{Applications with Data}\label{dataApp}
\subsection{Data}\label{sec51}
As in \cite{oldDCA}, main results are applied to annual data from the S\&P Composite Index from 1871 to 2020. The index consists of large US companies, and each company's weight is proportional to its market capitalization. The data was taken from \url{http://www.econ.yale.edu/~shiller/data.htm} and is collected for easy access at \url{https://github.com/HaydenBrown/Investing}. 

The number of companies in the US stock market has increased significantly from 1871 to 2020. In order to account for this increase, the data is split into three indexes, each covering a different time interval. The indexes are Cowles and Associates from 1871 to 1926, Standard \& Poor 90 from 1926 to 1957 and Standard \& Poor 500 from 1957 to 2020. Cowles and Associates index is a backward extension of the S\&P 90. The S\&P 90 consists of 90 companies, and the S\&P 500 consists of 500 companies.
\begin{table}[h]
\begin{center}
\caption{Data variable descriptions}
\label{table1}
\begin{tabular}{ |c|l| } 
\hline
\textbf{Notation} & \textbf{Description} \\
 \hline
 I & average monthly close of the S\&P composite index \\ 
 \hline
 D & dividend per share of the S\&P composite index \\ 
 \hline
 C & January consumer price index \\ 
 \hline
\end{tabular}
\end{center}
\end{table}

As in \cite{oldDCA}, the data is transformed so that annual returns incorporate dividends and are adjusted for inflation. In particular, returns are computed using the consumer price index, the S\&P Composite Index price and the S\&P Composite Index dividend. Use the subscript $n$ to denote the $n$th year of $C$, $S$ and $D$ from Table \ref{table1}. The return for year $n$ is computed as $\frac{I_{n+1}+D_n}{I_n}\cdot\frac{C_n}{C_{n+1}}$.

\subsection{Set-up}\label{setup}
In order to apply main results, the assumption that S\&P annual log-returns arise from a L\'evy alpha-stable process needs to be justified. This is done in two parts: first by establishing independence of annual log-returns, and second by establishing that annual log-returns follow a L\'evy alpha-stable distribution.

Independence of S\&P annual returns is established in \cite{oldDCA}. The L\'evy alpha-stable distribution of log-returns is established by visual inspection of quantile-quantile plots. Parameters are estimated using the iterative Koutrouvelis regression method with initial estimate $\mathcal{S}(1.91,1,.115,.0658)$ \cite{koutrouvelis1980regression,koutrouvelis1981iterative}. The resulting estimate for parameters is $\mathcal{S}(1.89,1.00,.110,.0658)$. Figure \ref{fig:qq} shows the quantiles of the estimated distribution vs the empirical quantiles of the log-returns. For comparison, Figure \ref{fig:qq} also shows the quantiles of the Normal distribution with mean $.0658$ and standard deviation $.169$ (used in \cite{oldDCA}) vs the empirical quantiles of the log-returns. 
\begin{figure}[H]
  \includegraphics[width=\linewidth]{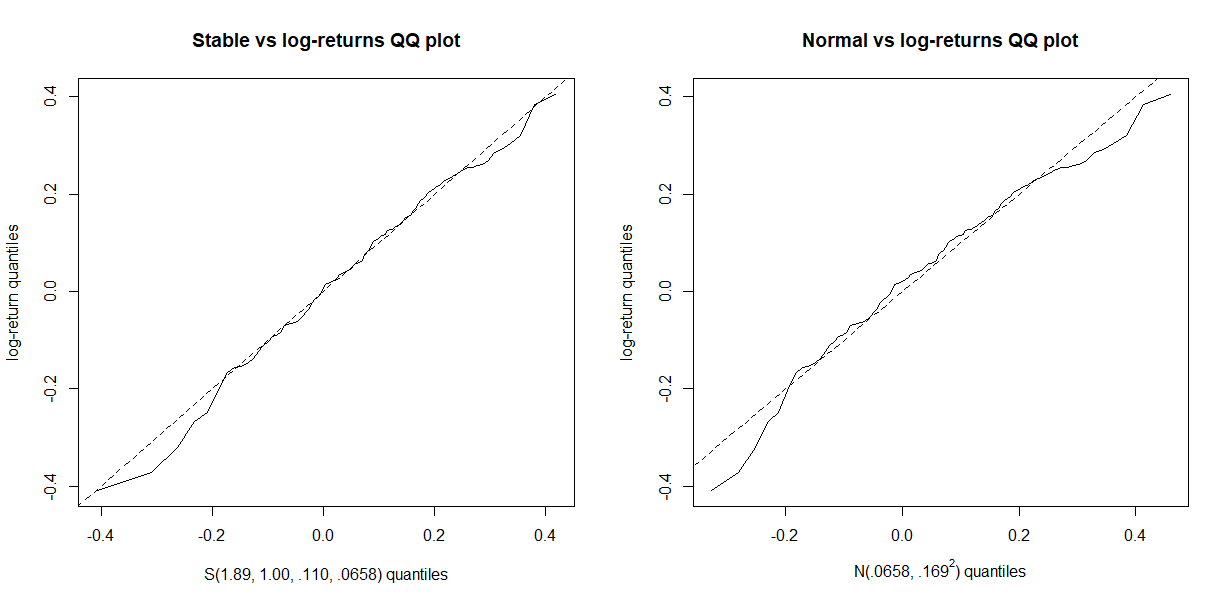}
  \caption{Left: The quantile-quantile plot of annual S\&P log-returns vs $\mathcal{S}(1.89,1.00,.110,.0658)$. Right: The quantile-quantile plot of annual S\&P log-returns vs $\mathcal{N}(.0658,.169^2)$.}
  \label{fig:qq}
\end{figure}
The distribution $\mathcal{S}(1.89,1.00,.110,.0658)$ offers a tighter quantile fit compared to the Normal distribution considered in \cite{oldDCA}. Note that the unit on the time domain is years, and applications using S\&P log-returns use only annual investment and withdrawals (i.e. $t_k=k$ for $k=0,1,...$).

\subsection{DCA quantiles}
In \cite{oldDCA}, $\mathcal{N}(.0658,.169^2)$ provided the best fit for annual log-returns. Section \ref{mainPI} generalizes the results of \cite{oldDCA}, allowing $\alpha<2$. In section \ref{setup}, $\mathcal{S}(1.89,1.00,.110,.0658)$ is given as a better fit for log-returns compared to $\mathcal{N}(.0658,.169^2)$. The goal here is to compare the lower bound of Theorem \ref{t1} when log-returns follow $\mathcal{S}(1.89,1.00,.110,.0658)$ versus $\mathcal{N}(.0658,.169^2)$. Figure \ref{fig:qqLB} shows the quantiles of the log lower bound when log-returns follow $\mathcal{S}(1.89,1.00,.110,.0658)$ versus $\mathcal{N}(.0658,.169^2)$.
\begin{figure}[H] 
  \includegraphics[width=\linewidth]{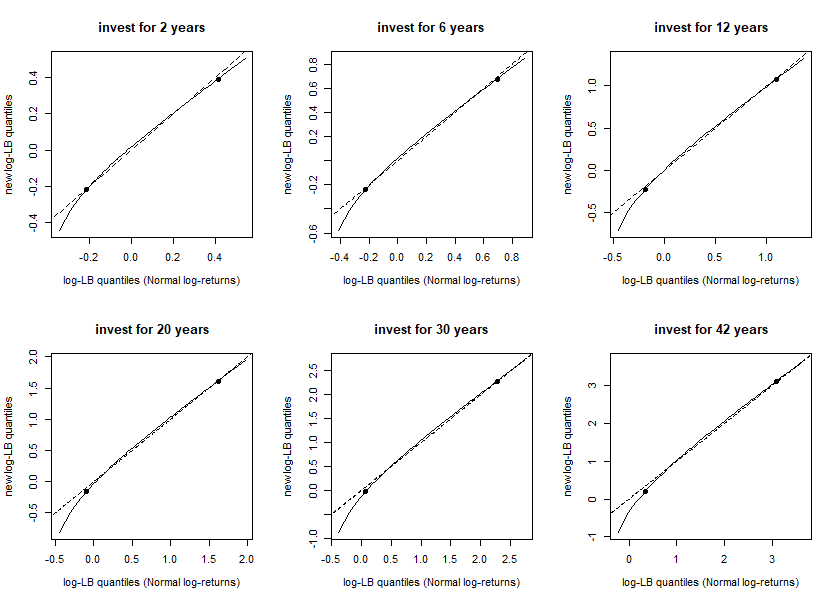}
  \caption{Quantiles .01 through .99 are computed for $\log Z_k$ of Theorem \ref{t1} using a range of investment lengths ($k=2,6,12,20,30,42$). The x-axis shows quantiles from \cite{oldDCA}, which uses $\mathcal{N}(.0658,.169^2)$ log-returns. The y-axis shows quantiles using $\mathcal{S}(1.89,1.00,.110,.0658)$ log-returns. The dashed line indicates where the curve should be if quantiles are identical. Points indicate the .05 and .95 quantiles.}
\label{fig:qqLB}
\end{figure}
In Figure \ref{fig:qqLB}, observe that quantiles are nearly identical except the quantiles under .05. When log-returns follow $\mathcal{S}(1.89,1.00,.110,.0658)$, the quantiles under .05 are less than when log-returns follow $\mathcal{N}(.0658,.169^2)$. Thus, taking into account the skewness and wider shape ($\alpha<2$) of log-returns leads to more significant downside risk compared to Normal log-returns. However, if a decision is based on quantiles .05 or greater, then the decision will likely be invariant to the choice of whether log-returns follow $\mathcal{S}(1.89,1.00,.110,.0658)$ versus $\mathcal{N}(.0658,.169^2)$.

\subsection{Discrete Withdrawls}
Algorithm \ref{alg1} is executed for $C=.95$, $\alpha=1.89$ and $\beta=1$. If $\mu\in\{.001n:\ n=1,...,500\}$, $k\in\{2,3,...,60\}$ and \eqref{ubms} holds, then in order to make $k$ equal withdrawals at each time $1,2,...,k$, with 95\% confidence, it is necessary for the initial investment to be at least $k$ times the amount of each withdrawal. 
\begin{equation}
\frac{\mu}{\sigma}\leq19.8070-0.7709\sqrt{k}+12.9209\sqrt[4]{k}-29.6668\sqrt[8]{k}.
\label{ubms}
\end{equation}
Figure \ref{ubms} shows that the bound \eqref{ubms} is very close to the $s_k$ returned by Algorithm \ref{alg1}. It has been verified that $\mu/\sigma=.066/.110$ satisfies \eqref{ubms} for $k=2,3,...,16$. Thus, in order to make $2\leq k\leq16$ equal, annual withdrawals from the S\&P Composite Index, with 95\% confidence, it is necessary for the initial investment to be at least $k$ times the amount of each withdrawal.
\begin{figure}[H]
  \includegraphics[width=\linewidth]{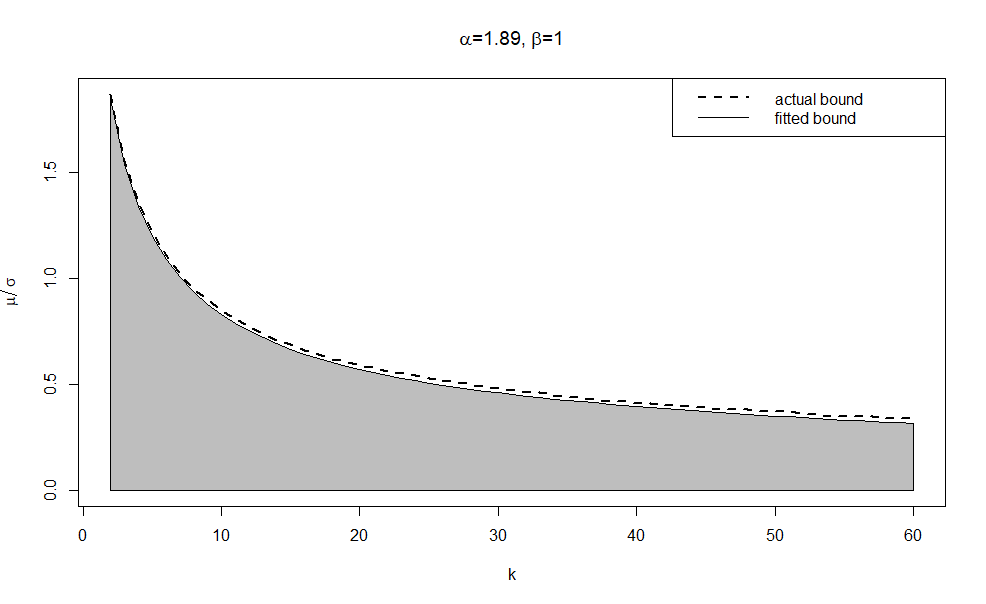}
  \caption{The line and shaded area illustrate \eqref{ubms} for $C=.95$. The dashed line is the returned $s_k$ from Algorithm \ref{alg1} using $C=.95$, $\alpha=1.89$ and $\beta=1$.}
  \label{fig:ubms}
\end{figure}

\subsection{Continuous Withdrawls}
This application is similar to what was done in section \ref{conW}. Suppose an investor wishes to continuously withdraw a total of 1 currency unit from the S\&P Composite Index over $n$ years, at a constant rate, with a given level of confidence. The same general procedure outlined in section \ref{conW} is used here. Let $C\in(0,1)$ denote the given level of confidence. Set $C=\mathbb{P}(Z^*\leq P)$, and solve for $P$, where $Z^*$ is as in Remark \ref{rcont}. Any initial investment less than $P$ results in a confidence level less than $C$. Set $\alpha=1.89$, $\beta=1$, $\sigma=.110\sqrt[1.89]{n}$ and $\mu=.0658n$. The goal here is to compute the necessary initial investment $P$ for $C\in\{.6+.001n:\ n=1,2,...,390\}$. Algorithm \ref{alg3} describes the procedure, and Figure \ref{fig:pc} illustrates $P$ as a function of $C$, for various $n$. To continuously withdraw $d\$$ with a high level of confidence, like 99\%, the initial investment needs to be at least $d\$$, when continuous withdraws are made over $n\leq 42$ years. For lower confidence, like 60\%, the initial investment needs to be at least $d/3\$$.

\begin{algorithm}
\caption{Compute $P$, given $C$}
\label{alg3}
\begin{algorithmic}
\Require $n\in\{2,3,...\}$, $Z\sim\mathcal{S}(1.89,-1,1,0)$, $C\in\{.6+.01n:\ n=1,2,...,35\}$
\State $\mu\gets.0658n$
\State $\sigma\gets.110\sqrt[1.89]{n}$
\State $\mathbb{P}(Z\leq q)\gets C$  \Comment{solve for $q$}
\State $P\gets\exp(\mu^*+\sigma^*q)$ \Comment{see Remark \ref{rcont} for $\mu^*,\sigma^*$}\\
\Return{$P$}
\end{algorithmic}
\end{algorithm}

\begin{figure}[H]
  \includegraphics[width=\linewidth]{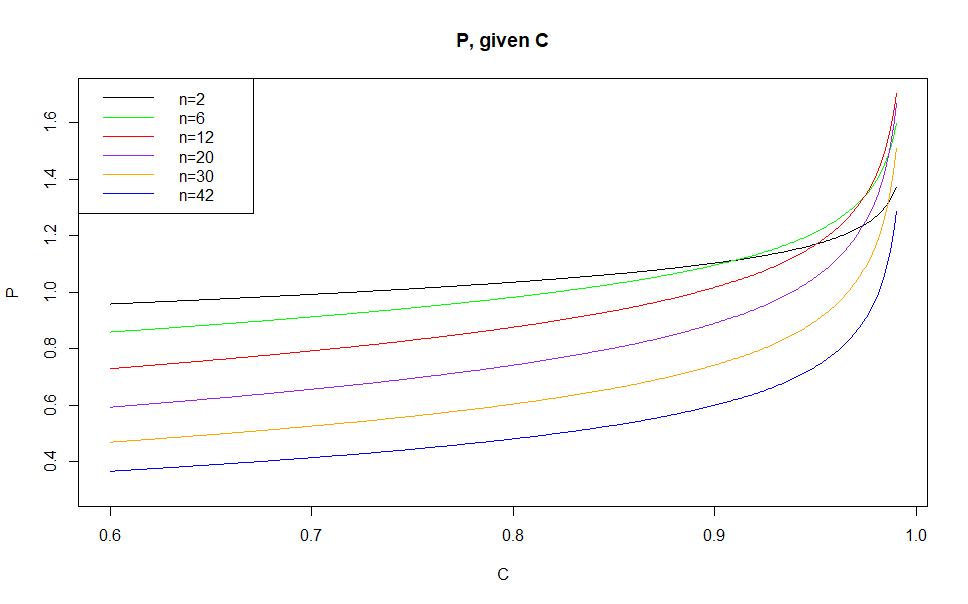}
  \caption{For each $C\in\{.6+.001n:\ n=1,2,...,390\}$, $P$ is computed using Algorithm \ref{alg3}. $P$ indicates a necessary initial investment in order to continuously withdraw 1 currency unit from the S\&P Composite Index over $n$ years, at a constant rate, with confidence level $C$.}
  \label{fig:pc}
\end{figure}

\section{Conclusions \& Further Research}\label{sec4}
As in \cite{oldDCA}, expected error between the lower bound of terminal wealth and the actual terminal wealth can be computed, but only in the case $\alpha>1$. When $\alpha\leq1$, the mean of a L\'evy alpha-stable random variable is undefined, causing the expected error to be undefined. An upper bound for the log-error between the lower bound of terminal wealth and the actual terminal wealth is also presented in \cite{oldDCA}. However, the L\'evy alpha-stable random variable has undefined 2nd or higher moments for $\alpha<2$. Consequently, the upper bound for log-error is less meaningful when $\alpha<2$ because higher moments cannot be used to tighten the upper bound.

The lower bound for terminal wealth of a regular investment schedule has interesting application for withdrawals, since a schedule of withdrawals can be transformed into a regular investment schedule. The result is an upper bound on the probability to complete a given schedule of withdrawals. Analogously, an upper bound for terminal wealth of a regular investment schedule would produce a lower bound on the probability to complete a given schedule of withdrawals. With a lower bound on the probability to complete a given schedule of withdrawals, it would be possible to specify a sufficient initial investment to complete a given schedule of withdrawals with a particular level of confidence. 

\begin{appendices}

\section{Proofs}\label{secA1}
\begin{lemma}
For all $a,c,x\in(0,\infty)$, 
\begin{equation*}
(a+c)\Big(\frac{x}{a}\Big)^{\frac{a}{a+c}}\leq x+c.
\end{equation*}
Moreover, there is equality at $x=a$.
\label{llb}
\end{lemma}
\begin{proof}
See the Appendix of \cite{oldDCA}.
\end{proof}

\begin{lemma}
Let $a\neq0$, $b\in\mathbb{R}$ and $A_i\sim\mathcal{S}(\alpha,\beta,\sigma^*_i,\mu^*_i)$ for $i=1,2$. In addition, let $A_1$ and $A_2$ be independent. Then
\begin{equation*}
aA_1+b\sim\begin{cases}
\mathcal{S}(\alpha,\frac{a}{\lvert a\rvert}\beta,\lvert a\rvert\sigma^*_1,a\mu^*_1+b)&\alpha\neq1\\
\mathcal{S}(1,\frac{a}{\lvert a\rvert}\beta,\lvert a\rvert\sigma^*_1,a\mu^*_1+b-\frac{2}{\pi}\beta a\sigma^*_1\log \lvert a\rvert)&\alpha=1
\end{cases},
\end{equation*}
and
\begin{equation*}
A_1+A_2\sim\mathcal{S}(\alpha,\beta,[(\sigma^*_1)^{\alpha}+(\sigma^*_2)^{\alpha}]^{\frac{1}{\alpha}},\mu^*_1+\mu^*_2).
\end{equation*}
\label{stableT}
\end{lemma}
\begin{proof}
The results are taken from \cite{nolan2020univariate}.
\end{proof}

\subsubsection*{Theorem \ref{tG}}
\begin{proof}
The result is established via induction. By Definition \ref{defz} $Z_1=Y_1=c_0X_1$. Since $\log X_1\sim\mathcal{S}(\alpha,\beta,\sigma\sqrt[\alpha]{t_1},\mu t_1)$, it follows from Lemma \ref{stableT} that 
\begin{equation*}
\log Z_1\sim\begin{cases}
\mathcal{S}(\alpha,\beta,c_0\sigma\sqrt[\alpha]{t_1},c_0\mu t_1)&\alpha\neq 1\\
\mathcal{S}(\alpha,\beta,c_0\sigma t_1,c_0\mu t_1-\frac{2}{\pi}\beta c_0\sigma t_1\log c_0)&\alpha= 1
\end{cases}.
\end{equation*} 

Now suppose the result holds for all $k'<k$ where $k\in\mathbb{N}\setminus\{1\}$. By \eqref{recursion}, 
\begin{equation}
Y_k=X_k\cdot(Y_{k-1}+c_{k-1}).
\label{zpart1}
\end{equation}
By Lemma \ref{llb}, 
\begin{equation}
(a_{k-1}+c_{k-1})\Bigg(\frac{Y_{k-1}}{a_{k-1}}\Bigg)^{\frac{a_{k-1}}{a_{k-1}+c_{k-1}}}\leq Y_{k-1}+c_{k-1},\quad w.p.1.
\label{zpart2}
\end{equation}
Combining \eqref{zpart1}, \eqref{zpart2} and Definition \ref{defz} yields $Z_k\leq R_k$ w.p.1.

Since $X_k$ and $Z_{k-1}$ are independent and log-stable with the same $\alpha$ and $\beta$, $Z_k$ is also log-stable. Applying properties of $\log$ to Definition \ref{defz},
\begin{equation}
\log Z_k=\log(a_{k-1}+c_{k-1})+\log X_k+b_{k-1}(\log Z_{k-1}-\log a_{k-1}).
\label{gzkk}
\end{equation}
Applying Lemma \ref{stableT} to \eqref{gzkk} shows that $\log Z_k\sim\mathcal{S}(\alpha,\beta,\sigma_k,\mu_k)$.
\end{proof}

\subsubsection*{Theorem \ref{t1}}
\begin{proof}
By Theorem \ref{tG}, the result holds when $\mu_k$ and $\sigma_k$ are as in Theorem \ref{tG}. It remains to be verified is that the expressions of $\mu_k$ and $\sigma_k$ given in Theorems \ref{tG} and \ref{t1} are equal. This is clearly the case for $k=1$. Suppose the expressions of $\mu_k$ and $\sigma_k$ given in Theorems \ref{tG} and \ref{t1} are equal for all $k'<k$ where $k\in\mathbb{N}\setminus\{1\}$. Then by Theorem \ref{tG} and the induction assumption,
\begin{equation*}
\begin{split}
\mu_k&=\log[\exp\mu_{k-1}+1]+\mu\\
&=\log\Big[\frac{\exp(\mu(k-1))-1}{\exp\mu-1}\exp\mu+1\Big]+\mu\\
&=\log\frac{\exp(\mu k)-1}{\exp\mu-1}+\mu.
\end{split}
\end{equation*}
Using Theorem \ref{tG}, the recursion $\sigma_j^{\alpha}=(b_{j-1}\sigma_{j-1})^{\alpha}+\sigma^{\alpha}$ with $\sigma_1^{\alpha}=(c_0\sigma)^{\alpha}t_1$ implies that $\sigma_k^{\alpha}=\sigma^{\alpha}(1+\sum_{j=1}^{k-1}(\prod_{i=1}^jb_{k-i})^{\alpha})$, where
\begin{equation*}
b_{k-i}=\frac{\exp\mu_{k-i}}{\exp\mu_{k-i}+1}=\frac{\exp\{\mu (k-i+1)\}-\exp\mu}{\exp\{\mu (k-i+1)\}-1}.
\end{equation*}
Observe that for all $i\in\mathbb{N}$,
\begin{equation*}
\frac{\exp(\mu k)-\exp\{\mu (i-1)\}}{\exp(\mu k)-1}\cdot\frac{\exp\{\mu (k-i+1)\}-\exp\mu}{\exp\{\mu (k-i+1)\}-1}=\frac{\exp(\mu k)-\exp(\mu i)}{\exp(\mu k)-1}.
\end{equation*}
It follows that 
\begin{equation*}
\prod_{i=1}^jb_{k-i}=\frac{\exp(\mu k)-\exp(\mu j)}{\exp(\mu k)-1}.
\end{equation*}
Thus, for all $k\in\mathbb{N}$,
\begin{equation}
\begin{split}
(\frac{\sigma_k}{\sigma})^{\alpha}&=1+\sum_{j=1}^{k-1}\Big(\frac{\exp(\mu k)-\exp(\mu j)}{\exp(\mu k)-1}\Big)^{\alpha}\\
&=1+\sum_{j=1}^{k-1}\Big(1-\frac{\exp(\mu j)-1}{\exp(\mu k)-1}\Big)^{\alpha}.
\end{split}
\label{cfms}
\end{equation}
\end{proof}

\subsubsection*{Theorem \ref{tcont}}
\begin{proof}
By Theorem \ref{t1}, for each $n\in\mathbb{N}$, there exists $Z_n$ such that $\frac{1}{n}Z_n\leq \mathcal{Y}_n$ w.p.1 and $\log Z_n\sim\mathcal{S}(\alpha,\beta,\sigma_n,\mu_n)$, using the substitutions $\mu\leftarrow\frac{\mu}{n}$ and $\sigma^{\alpha}\leftarrow\frac{\sigma^{\alpha}}{n}$ to evaluate $\mu_n$ and $\sigma_n$. To be more clear, the mean and scale parameters are 
\begin{equation*}
\begin{split}
\mu_n&=\frac{\mu}{n}+\log\frac{\exp\mu-1}{\exp\frac{\mu}{n}-1}\\
\sigma_n^{\alpha}&=\frac{\sigma^{\alpha}}{n}\Bigg(1+\sum_{j=1}^{k-1}\Big(1-\frac{\exp\frac{\mu j}{n}-1}{\exp\mu-1}\Big)^{\alpha}\Bigg).
\end{split}
\end{equation*}
It follows that $\mathbb{P}(\frac{1}{n}Z_n\leq x)\geq\mathbb{P}(\mathcal{Y}_n\leq x)$ for each $x\in\mathbb{R}$ and $n\in\mathbb{N}$. 

Note that $\lim_{n\to\infty}\mathcal{Y}_n$ is the integral of a geometric L\'evy alpha-stable process w.r.t. time (see \cite{milevsky2003continuous} for the derivation using $\alpha=2$ and $\beta=0$), so $\lim_{n\to\infty}\mathbb{P}(\mathcal{Y}_n\leq x)$ exists. Existence of the limit follows because sample paths of a geometric L\'evy alpha-stable process are c\`adl\`ag, which implies each sample path is integrable. Thus, $\mathcal{Y}_n$ converges pointwise on $\Omega$. Since each $\mathcal{Y}_n$ is $\mathcal{F}$-measurable, it follows from closure under pointwise limits that $\lim_{n\to\infty}\mathcal{Y}_n$ is $\mathcal{F}$-measurable. 

By L\'evy's convergence theorem, existence of $\lim_{n\to\infty}\mathbb{P}(\frac{1}{n}Z_n\leq x)$ will follow if the characteristic functions of $\log(\frac{1}{n}Z_n)$ converge. Thus, it suffices to show that $\sigma_n$ and $\mu_n-\log n$ converge.

An application of L'Hopital's rule reveals that
\begin{equation*}
\lim_{n\to\infty}\mu_n-\log n=\log\frac{\exp\mu-1}{\mu}.
\end{equation*}
By \eqref{cfms},
\begin{equation*}
\begin{split}
\sigma_n^{\alpha}&=\frac{\sigma^{\alpha}}{n}\Bigg[1+\sum_{j=1}^{n-1}\Big(\frac{\exp\mu -\exp(\mu\frac{j}{n})}{\exp\mu-1}\Big)^{\alpha}\Bigg]\\
&=\frac{\sigma^{\alpha}}{n}\Bigg[1+\sum_{j=1}^{n-1}\Big(\frac{1-\exp(-\mu\frac{j}{n})}{1-\exp(-\mu)}\Big)^{\alpha}\Bigg].
\end{split}
\end{equation*}
Therefore
\begin{equation*}
\begin{split}
\lim_{n\to\infty}\sigma_n^{\alpha}&=\sigma_n^{\alpha}\lim_{n\to\infty}\sum_{j=1}^{n-1}\frac{1}{n}\Big(\frac{1-\exp(-\mu\frac{j}{n})}{1-\exp(-\mu)}\Big)^{\alpha}\\
&=\sigma^{\alpha}\int_0^1\Big(\frac{1-\exp(-\mu x)}{1-\exp(-\mu)}\Big)^{\alpha}dx.
\end{split}
\end{equation*}
\end{proof}

\subsubsection*{Proof of the upper bound in \eqref{xxr}}
Let $f(x)=1-\exp(-\mu x)$ and $U(x)=\big(1-\exp(-\mu)\big)x^r$. Recall that $r=\frac{\mu}{\exp\mu-1}$. The goal is to show $U(x)-f(x)\geq0$ on $[0,1]$. The result clearly holds for $\mu=0$, so require $\lvert\mu\rvert>0$. Observe that $U(x)-f(x)$ is differentiable on $(0,1]$, where
\begin{equation}
\begin{split}
\frac{d}{dx}[U(x)-f(x)]&=\big(1-\exp(-\mu)\big)rx^{r-1}-\mu\exp(-\mu x)\\
&=\frac{\mu}{\exp\mu}[x^{r-1}-\exp(\mu(1-x))].
\end{split}
\label{ufx}
\end{equation}

From \eqref{ufx}, it follows that $\frac{d}{dx}[U(x)-f(x)]\odot0$ iff $\exp\frac{\mu(1-x)}{r-1}-x\odot0$, where $\odot\in\{\leq,\geq\}$. Set $g(x)=\exp\frac{\mu(1-x)}{r-1}-x$. Observe that $g(1)=0$ and $g(x)$ is continuous and strictly convex on $(0,1]$. Thus, either $g(x)\leq0$ on $(0,1]$, $g(x)\geq0$ on $(0,1]$, or there exists $y\in(0,1)$ such that $g(x)\geq0$ on $(0,y]$ and $g(x)\leq0$ on $(y,1]$. Accounting for the fact that $\lim_{x\to0^+}U(x)-f(x)=U(0)-f(0)$, all three cases imply that the global minimum of $U(x)-f(x)$ occurs at $x=0$ or $x=1$. Last, observe that $U(x)-f(x)=0$ at $x=0$ and $x=1$. 

\subsubsection*{Theorems \ref{tG*}, \ref{t1*} and \ref{tW}}
\begin{proof}
Observe that $\log(X_k^{-1})\sim\mathcal{S}(\alpha,-\beta,\sigma\sqrt[\alpha]{t_k-t_{k-1}},-\mu(t_k-t_{k-1}))$. For Theorems \ref{tG*} and \ref{t1*}, apply Theorems \ref{tG} and \ref{t1} to \ref{recursionW2}. Theorem \ref{tW} follows from the fact that $P^*(\omega)\leq P\iff W_k(\omega)\geq0$ and $Z_0^*\leq P^*$ w.p.1. See \eqref{recursionW2} for the construction of $P^*$ and \eqref{recursionW} for the construction of $W_k$. 
\end{proof}

\end{appendices}


\bibliography{sn-bibliography}


\end{document}